\documentclass[submission,copyright,creativecommons]{eptcs}

\providecommand{\event}{TERMGRAPH 2018}   
\usepackage{underscore}                   
\usepackage{graphicx}
\usepackage{mathrsfs}
\usepackage{stmaryrd}
\usepackage{subfigure}
\usepackage{url}
\usepackage{multirow}
\usepackage{wrapfig}
\usepackage{booktabs} 
\usepackage{amssymb,amsmath,amsthm}
\usepackage[ruled, vlined]{algorithm2e}
\usepackage[abs]{overpic}
\usepackage{syntax}
\usepackage{framed}
\usepackage{fixltx2e}

\newcommand{\Porgy}{\textsc{Porgy}\xspace}

\newcommand{\Tulip}{\textsc{Tulip}\xspace}

\newcommand{\Node}{{\tt node}\xspace}
\newcommand{\Edge}{{\tt edge}\xspace}
\newcommand{\Port}{{\tt port}\xspace}
\newcommand{\iid}{{\tt id}\xspace}

\newcommand{\iter}{{\tt iter}\xspace}
\newcommand{\visit}{{\tt visit}\xspace}
\newcommand{\NotNode}{{\tt NotNode()}\xspace}

\newcommand{\tmatch}[1]{{\tt match(}#1{\tt )}}

\newcommand{\one}[1]{{\tt one(}#1{\tt )}}

\newcommand{\while}[1]{{\tt while(}#1{\tt )}}
\newcommand{\doo}[1]{{\tt do(}#1{\tt )}}

\newcommand{\whiledoo}[2]{{\tt while(}#1{\tt )do(}#2{\tt )}}

\newcommand{\repeatt}[1]{{\tt repeat(}#1{\tt )}}

\newcommand{\Ra}{\Rightarrow}

\newcommand{\lra}{\longrightarrow}

\newcommand{\F}{{\cal F}}

\newcommand{\LL}{{\cal L}}

\newcommand{\saturated}{\textit{saturated}}

\ifx\thm\undefined
\newtheorem{thm}{Theorem}
\fi

\ifx\exmp\undefined
\newtheorem{exmp}[thm]{Example}
\fi

\ifx\proposition\undefined
\newtheorem{proposition}[thm]{Proposition}
\fi

\ifx\lemma\undefined
\newtheorem{lemma}[thm]{Lemma}
\fi

\ifx\definition\undefined
\newtheorem{definition}[thm]{Definition}
\fi

\title{Finding the Transitive Closure of Functional Dependencies using Strategic Port Graph Rewriting}
\author{
    J{\'a}nos Varga
    \institute{Department of Informatics}
    \institute{King's College London}
    \email{janos.varga@kcl.ac.uk}
}

\def\titlerunning{Finding Transitive Closure using Strategic Port Graph Rewriting}
\def\authorrunning{J. Varga}

\SetAlgorithmName{Strategy}{strategy}{List of Strategies}
\theoremstyle{definition}
\newtheorem{property}[thm]{Property}

\begin{document}

\maketitle

\begin{abstract}
We present a new approach to the logical design of relational databases, based on strategic port graph rewriting. 
We show how to model relational schemata as attributed port graphs and provide port graph rewriting rules to perform computations on functional dependencies.
Using these rules we present a strategic graph program to find the transitive closure of a set of functional dependencies. This program is sound, complete and terminating, assuming that there are no cyclical dependencies in the schema.
\end{abstract}

\section{Introduction}
Traditionally, steps of relational database design include Conceptual, Logical and Physical modelling. 
The theory behind these steps 
is well-understood and is part of the syllabus of many databases courses. Yet, database professionals often consider Logical Design (normalisation) too cumbersome and do not apply normalisation theory, despite the clear advantages of normalised database designs. Badia and Lemire~\cite{CallToArms} also highlight that conceptual and logical models do not always carry enough information about database semantics thereby leading the architect to a sub-optimal design. 

We use \emph{attributed port graphs} to represent a relational schema and its semantics. Port graphs are graphs where edges are connected to nodes at specific points, called ports. 
In port graphs, nodes, edges and ports can have attributes, which are used to represent properties of the system modelled. In this paper we focus on using port graphs in the logical phase of database design. We show that port graphs are a good choice of data structure to store relational metadata and can be transformed without the loss of metadata. 

To specify the transformations applied to relational schemata, we use \textit{port graph rewriting systems}, a general class of graph rewriting systems~\cite{Courcelle90}. The implementation framework we use is \Porgy~\cite{FernandezKP18} -- a visual, interactive tool for the specification, simulation and analysis of systems based on port graph rewriting. \Porgy provides a graphical interface, where users can define a system and specify its dynamics by means of port graph rewrite rules and strategies. Port graphs have node, port and edge \emph{attributes}, whose values are taken into account in port graph morphisms (used to define rewriting steps) and in strategy expressions (to control the application of rules).
\textit{Strategic graph programs}~\cite{FernandezKP18}, consisting of an initial port graph and a set of rewrite rules controlled by a strategy, are the essence of \Porgy. The strategy language offers separate primitives to select subgraphs of the model as focusing positions for rewriting and to select the rewrite rules to be applied, following the separation of concerns principle which makes programs easier to maintain and adapt. The strategy language also allows users to define strategies using not only operators to combine graph rewriting rules but also operators to deal with graph traversal and management of rewriting positions in a graph. \Porgy provides a visual representation of the set of rewrite derivations (a \emph{derivation tree}) and includes features such as cycle detection, to facilitate debugging.

We extend the rule language by adding the possibility to specify application conditions for a rule. That is, as part of a rewrite step, in the rule editor, users can define a set of conditions which are evaluated after a morphism has been found. The rewrite step is applied only if the rule condition evaluates to true. As a use-case we provide a set of port graph rewriting rules and a strategy to calculate the transitive closure of a set of functional dependencies. Although there are a number of tools already available to do the same, a distinctive advantage of our implementation is that it is visual and backtrackable (thanks to the derivation tree feature of \Porgy). Our strategy is sound, complete and terminating, given the restriction that there are no cyclical dependencies in the schema.

Summarising, our contributions are:
\begin{enumerate}
    \item a new visual language, based on port graphs, for logical design of relational schemata,
    \item generic application conditions for rules (a port graph rewriting language extension),
    \item a strategic graph program to find the transitive closure of a set of functional dependencies.
\end{enumerate}
This last contribution is a key step towards building strategic graph programs to find minimal covers, candidate keys and Third Normal Form (3NF) relation schemata.

\emph{Related Work.}
Graph theory and graph rewriting is by no means a new addition to the set of tools that have been used for relational database design. In~\cite{EmbleyM11} hypergraphs are used and their well-formed property (called a canonical hypergraph) determines the quality of design they represent. The authors of \cite{SaiedianS96} used directed graphs to find all candidate keys of a relation in polynomial time. A special family of labelled graphs, FD-graphs, were introduced in \cite{AusielloDS83} to obtain meaningful closures of a set of functional dependencies. In terms of graph transformations and rewriting we highlight two works. Hypergraph rewriting was used in \cite{BatiniD78} for the manipulation of functional dependencies and Triple Graph Grammars were used in \cite{JahnkeZundorf99} to optimize an already existing schema. Section~\ref{sec:portgraphdb} translates the work presented in~\cite{AusielloDS83} to port graphs and extends it.

\emph{Organization.}
This paper is organised as follows. 
We briefly review relational database theory and port graph rewriting background in Section~\ref{sec:background}. 
We present our port graph visual language for logical design of relation schemata in Section~\ref{sec:portgraphdb}.
In Section~\ref{sec:rulecond} we define the syntax of the language for generic rule application conditions.
Section~\ref{sec:tcstrat} illustrates how the visual language and the generic rule application conditions can be used to find the transitive closure of a set of FDs.
We then conclude in Section~\ref{sec:conclusion} by highlighting how these results can be used in future work.

\section{Background}
\label{sec:background}
In this section we will review the definitions and background in relational databases and port graph rewriting that we are going to use throughout this paper. 
Due to space constraints, for formal definitions and proofs, this section will refer the reader to the relevant works rather than recalling them. 
We will also briefly review related work that used graphs to represent or transform relational schemata.

\subsection{Relational Database Design}
We assume that the reader is familiar with the theory of logical design of relational databases~\cite{Codd70,Codd71a}. 
In particular, the definitions of: \textit{relation schema}, \textit{attribute}, \textit{candidate key} and \textit{functional dependency} (FD)~\cite{Codd70,Codd71a}. 
We refer to a single attribute with letters from the beginning of the alphabet $A, B, \ldots$ and to attribute sets with letters from the end of the alphabet $X, Y, Z$. 
This paper will denote the set of all FDs of a relation schema $\Sigma_R$ or just $\Sigma$, where appropriate.  
We also assume familiarity with the inference rules of functional dependencies, also known as Armstrong's Axioms~\cite{BeeriFH77}: Transitivity, Trivial Dependency, Augmentation, Union, Decomposition, Pseudotransitivity. 
We call the set of all FDs that can be inferred from $\Sigma$, using Armstrong's Axioms, the \textit{syntactic closure} $\Sigma^+$. 
It was shown that Armstrong's Axioms are sound and complete, which means that they find only and all (respectively) semantically correct dependencies. 
This work assumes that a) FDs are in canonical form ie. only single attributes appear on the right-hand sides of the FDs and b) there are no cyclical dependencies.


\subsection{Port Graph Rewriting}
\label{subsec:pg}
An attributed port graph is a labelled attributed graph where nodes have specific points-of-connection called ports, and edges that are attached to ports. In this subsection we recall the most important port graph rewriting constructs from Sections 2 and 3 of~\cite{FernandezKP18}, where the full formal definitions can also be found.

A \emph{port graph rewrite rule} is a port graph $L \Ra_C R$ consisting of two subgraphs $L$ and $R$ together with an \emph{arrow} node that links them. 
Each rule is characterised by its arrow node, which has a unique name (the rule's label), a condition \emph{Where} restricting the rule's matching, 
and ports to control the rewiring operations when rewriting steps are computed. 
Edges that run between ports of $L$, $R$ and the arrow node are coloured red by \Porgy to distinguish them from normal edges. 
We recall the definition of the \emph{matching morphism} that states that a \emph{match} $g(L)$ of the left-hand side is found in $G$ 
if there is a total port graph morphism $g$ from $L$ to $G$ such that if the arrow node has an attribute \emph{Where} with value $C$, then $g(C)$ is true in $G$. 
$C$ is of the form $\saturated(p_1) \wedge ... \wedge \saturated(p_n) \wedge B$. 
The predicate $\saturated(g(p_i))$ is true if there are no edges between $g(p_i)$ and ports outside $g(L)$ in $G$ -- this ensures that no edges will be left dangling in rewriting steps.
$B$ is a Boolean expression such that all its variables occur in $L$.
To aid visual design of rewrite rules, \Porgy allows us to name nodes in $L$ and $R$ but these are treated by the system as \emph{node name variables}.
This means that node name variables identify the nodes on both sides of the rule but are instantiated when $g(L)$ is found using actual values from the matching nodes.

Our contribution to the rewrite rule language is to provide the functionality of generic application conditions. This task was two-fold: we created a grammar for $B$ and implemented a \Porgy Rule Editor plug-in (called Rule Conditions).

We also recall here that a \textit{strategic graph program} consists of a \emph{located graph} (a port graph with two distinguished subgraphs that specify the locations where rewriting should take place or not), a set of \emph{located rewrite rules}, and a \emph{strategy expression}. In a located graph $G_{P}^Q$, $P$ represents the \emph{position} subgraph of $G$ where rewriting steps may take place and $Q$ represents the \emph{banned} subgraph of $G$ where rewriting steps are forbidden. A located rewrite rule $L \Ra_C R_{M}^N$ can update $P$ and $Q$ in a rewrite step such that $P'=(P \setminus g(L)) \cup g(M)$ and $Q'= (Q \setminus g(L)) \cup g(N)$.

Our work to find the transitive closure is implemented in the form of a strategic graph program.

\subsection{Abstract Reduction Systems}
\label{subsec:ars}
We use the theoretical framework of Abstract Reduction Systems (ARS)~\cite[Chapter~2]{termrewrite} to prove termination. 
Various techniques to construct termination proofs have been published and we shall recall one here.
This technique requires the embedding of the ARS $(A, \lra)$ into another ARS $(B, >)$ of which we know that it terminates. 
Our choice, $(\mathbb{N}, >)$ terminates because every descending chain $a_0 > a_1 > \ldots$ is finite.

\begin{definition}[Monotone mapping]
The mapping $\varphi : A \lra B$ is monotone if $x \lra x' \Rightarrow \varphi(x) > \varphi(x')$. $\varphi$ is also known as the \emph{measure function}.
\end{definition}

It does not automatically follow from the above that to prove termination of an abstract reduction system 
it suffices to find a measure function $\varphi$ to embed the system into, for example, $(\mathbb{N}, >)$. 
But if we can prove that the ARS is finitely branching then we can make use of the following lemma:
\begin{lemma}
	\label{lemma:terminates}
	A finitely branching reduction terminates iff there is a monotone embedding into $(\mathbb{N}, >)$.
\end{lemma}

\section{A Visual Language for Relational Schema Design}
\label{sec:portgraphdb}
We now show how relational schemata (using functional dependencies only) can be modelled as attributed port graphs. We use the $'.'$ (dot, member-of) operator to refer to a particular port of a node.

We define relation schema attributes and FDs as nodes. The fact that an attribute belongs to the right- or left-hand side of a FD is represented by edges. However, when adding FDs to the visual language, we face a challenge. Because of the semantics of a FD (i.e. LHS determines RHS), strategic graph programs executed on this visual language have to be able to distinguish between LHS and RHS attributes of FDs. A non-trivial FD must have at least one attribute on both sides where (RHS $\not\subseteq$ LHS). Also, as per the separation of concerns principle, a FD has to be aware of the list of attributes on its sides, not the other way around. Formally, we say that:

\begin{definition}[Functional Dependency Port Graph, FDPG]
\label{def:fdpg}
Let $R$ be a relation schema and $\Sigma$ its set of functional dependencies.
A {\em Funtional Dependency Port Graph} representing $\Sigma$ is an attributed port graph~\cite{FernandezKP18} $G_\Sigma=( V,P,E,D )_{\F}$ and is defined as:
\begin{itemize}
	\item $V = V_A \cup V_{FD}$ is a union of two disjoint sets of nodes: 
	\begin{itemize}
		\item $V_A$: set of Attribute nodes, one node for every attribute in $R$;
		\item $V_{FD}$: set of Functional Dependency nodes, one node for every functional dependency in $\Sigma$;
	\end{itemize}
	\item $P = P_A \cup P_{FD}$ is a union of two defined sets of ports: 
	\begin{itemize}
		\item $P_A = \{pFD\}$ and 
		\item $P_{FD} = \{pFDLHS, pFDRHS\}$;
	\end{itemize}
	\item $E$ is a finite set of edges between ports; two ports may be connected by only one edge;
	\item $D$ a set of records~\cite{FernandezKP18};
\end{itemize}
and a set ${\F}$ of functions $Connect$, $Attach$ and $\LL$ such that: 
\begin{itemize}
	\item Connect: for each edge $e \in E$, $Connect(e)$ is the pair $(p_1,p_2)$ of ports connected by $e$ where the only allowed pairs are (pFD, pFDLHS) and (pFDRHS, pFD). 
			For every dependency $\varphi \in \Sigma: X \to A$ the pFD port of every attribute node corresponding to $X$ will be connected to the pFDLHS port of the dependency node corresponding to $\varphi$ and the pFDRHS port of the FD node $\varphi$ will be connected to the pFD port of the attribute node representing $A$.
	\item Attach: 
	\begin{itemize}
		\item for each port $p \in P_A$, $Attach(p)$ is the node $n \in V_A$ to which the port belongs;
		\item for each port $p \in P_{FD}$, $Attach(p)$ is the node $n \in V_{FD}$ to which the port belongs;
	\end{itemize}
	\item $\LL$ a labelling function~\cite{FernandezKP18}.
\end{itemize}
\end{definition}

The following properties directly follow from Definition~\ref{def:fdpg}.
\begin{property}[Cardinality of set $V$]
  In a FDPG $G_\Sigma=( V,P,E,D )_{\F}$, $|V| = |V_A| + |V_{FD}| = |R| + |\Sigma|$.
  \begin{proof}
    The mappings $R \to V_A$ and $\Sigma \to V_{FD}$ are bijections.
  \end{proof}
\end{property}

\begin{property}[Cardinality of set $E$]
  Given the set of functional dependencies $\Sigma = \{\varphi_1,\ldots,\varphi_k\}$ and a FDPG $G_\Sigma=( V,P,E,D )_{\F}$,
  $|E| = |\Sigma| + \sum_{i=1}^{k} |LHS(\varphi_i)|$.
  \begin{proof}
    Number of edges = one right-hand side edge per dependency + sum of the sizes of the left-hand side of each functional dependency.
  \end{proof}
\end{property}

We implement FDPGs in PORGY. Firstly, using the set $D$ of records, we introduce an attribute called \emph{RelDbType} which denotes the role of the node in the relational context. Every new node and port created in a FDPG-based logical model have to have a constant \emph{RelDbType} value, placing it in the appropriate set of $V_A, V_{FD}, P_A$ or $P_{FD}$. Attribute nodes have \emph{RelDbType} = ATTR, FD nodes have \emph{RelDbType} = FD. The port of an attribute that handles the connection to either side of a FD has \emph{RelDbType} = pFD. The LHS and RHS connection ports of a FD node have \emph{RelDbType} = FDLHS and FDRHS, respectively. Both FDLHS and FDRHS ports have an integer attribute \emph{FunctionalArity} defined that allows the system to store the number of attributes on each side. This is required because the matching algorithm does not enforce exact arity since a particular port can be connected to other ports outside the match found, however, when matching on FDs and their LHSs, every single LHS attribute has to be in the matching subgraph.

\begin{figure}[h]
    \centering
    \begin{overpic}[scale=0.7]{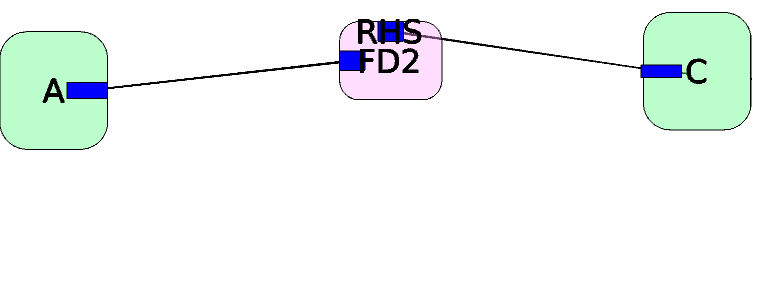}
        \texttt{
            \put(0,40) { {\parbox{2in} { \scriptsize{ 
                \underline{node:}\\
                viewLabel="A"\\
                RelDbType="ATTR"\\ 
                \underline{port:}\\
                RelDbType="pFD" }}}}
            \put(180,40) { {\parbox{2in} { \scriptsize{ 
                \underline{node:}\\
                viewLabel="FD2"\\
                RelDbType="FD"\\ 
                \underline{LHS port:}\\
                viewLabel="LHS"\\
                RelDbType="FDLHS"\\
                FunctionalArity=1\\
                \underline{RHS port:}\\
                viewLabel="RHS"\\
                RelDbType="FDRHS"\\
                FunctionalArity=1 }}}}
            \put(335,50) { {\parbox{2in} { \scriptsize{ 
                \underline{node:}\\
                viewLabel="C"\\
                RelDbType="ATTR"\\ 
                \underline{port:}\\
                RelDbType="pFD" }}}}
        }
    \end{overpic}
    \caption{The functional dependency $A \to C$.}
    \label{fig:fdpgexample}
\end{figure}

Both FD and ATTR nodes have an integer UID attribute that allows the rules to assign a numeric identity value to them. This is useful when a rule adds a new FD and we want control over the value of the unique identifier. We make use of the built-in \emph{viewLabel} attribute to assign meaningful node name constants (in models) and variables (in rules) to nodes.
Figure~\ref{fig:fdpgexample} shows a functional dependency $A \to C$ as a port graph, with the relevant attribute values.

Based on the FD-path definition of Ausiello et al.~\cite{AusielloDS83} we define an FDPG-path as follows:
\begin{definition}[FDPG-Path]
\label{def:fdpgpath}
Given an FDPG $G_\Sigma=(V,P,E,D)_{\F}$, an attribute node set $X \subseteq V_A$ and an attribute node $j \in V_A$, a (directed) FDPG-Path $\langle X, j \rangle$ from $X$ to $j$ is a minimal subgraph $G_\Sigma' = (V',P',E',D')_{\F'}$ of $G_\Sigma$ such that $X \cup \{j\} \subseteq V_A'$ and one of the following conditions holds:
\begin{enumerate}
	\item there exists $v \in V_{FD}'$ such that for all $x_i \in X, (x_i.pFD, v.pFDLHS) \in E'$ and $(v.pFDRHS, j.pFD) \in E'$, 
	i.e. there exists $v \in V_{FD}'$ such that $X$ is the set of all left-hand side attributes and $j$ is the right-hand side attribute of the functional dependency represented by $v$;
	\item there exist $v \in V_{FD}'$ and $k \in V_A'$ such that $(k.pFD, v.pFDLHS) \in E'$ and $(v.pFDRHS, j.pFD) \in E'$ and there is an FDPG-Path $\langle X, k \rangle$ included in $G_\Sigma'$;
	\item there exist $v \in V_{FD}'$ and $K \subseteq V_A'$ such that for all $k_i \in K, (k_i.pFD, v.pFDLHS) \in E'$ \\
	and $(v.pFDRHS, j.pFD) \in E'$ and $n$ FDPG-Paths $\langle X, k_1\rangle, \ldots \langle X, K_n \rangle$ are included in $G_\Sigma'$.
\end{enumerate}
\end{definition}

\begin{property}[Transitivity]
  \label{prop:trans}
  Condition 2 of Definition~\ref{def:fdpgpath} represents Armstrong's Transitivity axiom.
  \begin{proof}
    An FDPG-Path $\langle X, k \rangle$ represents FD $X \to k$. \\
	Edges $(k.pFD, v.pFDLHS) \in E$ and $(v.pFDRHS, j.pFD) \in E$ represent FD $k \to j$. 
	The fact that there is an FDPG-Path $\langle X,j \rangle$ means the FD $X \to j$ exists.
  \end{proof}
\end{property}

\begin{property}[Union]
  \label{prop:union}
  Condition 3 of Definition~\ref{def:fdpgpath} represents Armstrong's Union axiom.
  \begin{proof}
	The $n$ FDPG-Paths $\langle X, k_1\rangle, \ldots, \langle X, k_n \rangle$ represent FDs $X \to k_1, \ldots, X \to k_n$.\\
	Edges $k_i \in K, (k_i.pFD, v.pFDLHS) \in E$ and $(v.pFDRHS, j.pFD) \in E$ represent FDs $k_1 \to j, \ldots, k_n \to j$. 
	The fact that $n$ FDPG-Paths $\langle X, j \rangle$ exist means that $X \to k_1, \ldots, X \to k_n$ can be unified into $X \to K$.
  \end{proof}
\end{property}

We now turn our attention to extending the port graph rewriting language so that strategic graph programs can be created to find the transitive closure of a Functional Dependency Port Graph.

\section{Generic Rule Application Conditions}
\label{sec:rulecond}
We recall the structure of the arrow node \emph{Where} attribute defined in Section~\ref{subsec:pg}. In this section we extend the rewrite rule language to provide the functionality of generic application conditions. Firstly, we define the context-free grammar for $B$, secondly, we present the \Porgy Rule Editor plug-in (called Rule Conditions).

The EBNF grammar for $B$ is defined in Figure~\ref{fig:rulecondgrammar}. The structure of the grammar was inspired by C++ and follows the operator precedence of C++, too.
\grammarindent1.5in
\begin{figure}[!htb]
    \centering
    \begin{framed}
    \begin{flushleft}
        \begin{grammar}
            <node> ::= '$n($' "valid LHS node id" '$)$'
            
            <edge> ::= '$e($' "valid LHS edge id" '$)$'
            
            <element attribute> ::= ( <node> | <edge> ) '.' "quoted_attribute_name"
    
            <factor> ::= "number" | "quoted_string" | '(' <expression> ')' | '$!$' <factor> 
            \alt <element attribute> | 'max(' <expression> ',' <expression> ')'
            \alt 'min('~<expression>~','~<expression>')' | 'random('<factor>')'
    
            <term> ::= <factor> \{('$*$' | '$/$' | '$\%$') <factor>\} 
            
            <expression> ::= <term> \{('$+$' | '$-$' ) <term>\} 
            
            <comp operator> ::= '$==$' | '$!=$' | '$\textgreater$' | '$\textless$' | '$\textgreater=$' | '$\textless=$'

            <comparison> ::= <expression> <comp operator> <expression> | 'NotNode('"quoted_attribute_name" <comp operator> <expression>')'
            
            <logical expression> ::= <logical term> \{ '||' <logical term> \}
            
            <logical term> ::= <logical factor> \{ '\&\&' <logical factor> \}
            
            <logical factor> ::= <comparison> | '$!$' <logical factor> | '(' <logical expression> ')'
                
            <rule condition> ::= \{<logical expression>\};
        \end{grammar}
    \end{flushleft}
    \end{framed}
    \caption{The rule application condition grammar.}
    \label{fig:rulecondgrammar}
\end{figure}

We point out that when referring to a \Node or \Edge the user has to use its internally assigned \iid. Also, due to the implementation of \Porgy, there is no \Port construct in the grammar -- they have to be referred to as \Node. This is because the underlying graph engine (\Tulip) processes ports as nodes. Terminal {\tt number} can be any integer or floating-point number and {\tt quoted_string} is an arbitrary-length string made up of letters, digits and symbols in double quotes. Similarly, {\tt quoted_attribute_name} is a valid name of an attribute of \Node, \Edge or \Port.

We highlight the \NotNode operator: it iterates all nodes of $G$ and checks if there exists a node with an attribute {\tt quoted_attribute_name} and if the comparison on them evaluates to true. Intuitively, if at least one such node is found in $G$, \NotNode returns false. It is very important to note here that this check is performed on the entire graph $G$, not just in $g(L)$. This is fundamentally different from the rest of the rule application condition grammar, which only applies to $g(L)$. This is a consequence of the definition of the port graph rewrite rule which states that all variables in the Boolean expression of the \emph{Where} attribute have to occur in $L$, so that the matching algorithm can work with them. When a match $g(L)$ is found, all variables of $B$ are mapped so that their actual values can be found. However, when checking the absence of a node, we are not constrained by this, because we are not specifying a node in \NotNode -- we are only specifying an attribute comparison that \emph{must} evaluate to false on all nodes of $G$.

\Porgy offers a modular plug-in system allowing developers to create Python/C++ plug-ins. We implemented an LL-parser for the above detailed context-free grammar in C++ using the Boost Spirit Parser Framework. This framework generates and executes the parser design-time and builds and evaluates an abstract syntax tree run-time. When evaluated, the Boolean result is ANDed to the rest of the arrow node \emph{Where} attribute by the matching algorithm. The parser ensures that all nodes, edges, ports and attributes referred to in the conditions exist on the LHS of the rule.

We also added a UI extension to \Porgy that allows users to specify and parse/check the rule conditions. A screenshot of PORGY with the Conditions editor is presented in Figure~\ref{fig:PorgyScreenshot}. Examples of rule conditions can be found in Section~\ref{sec:tcstrat}.
\begin{figure}[!htb]
    \centering
    \includegraphics[scale=0.55]{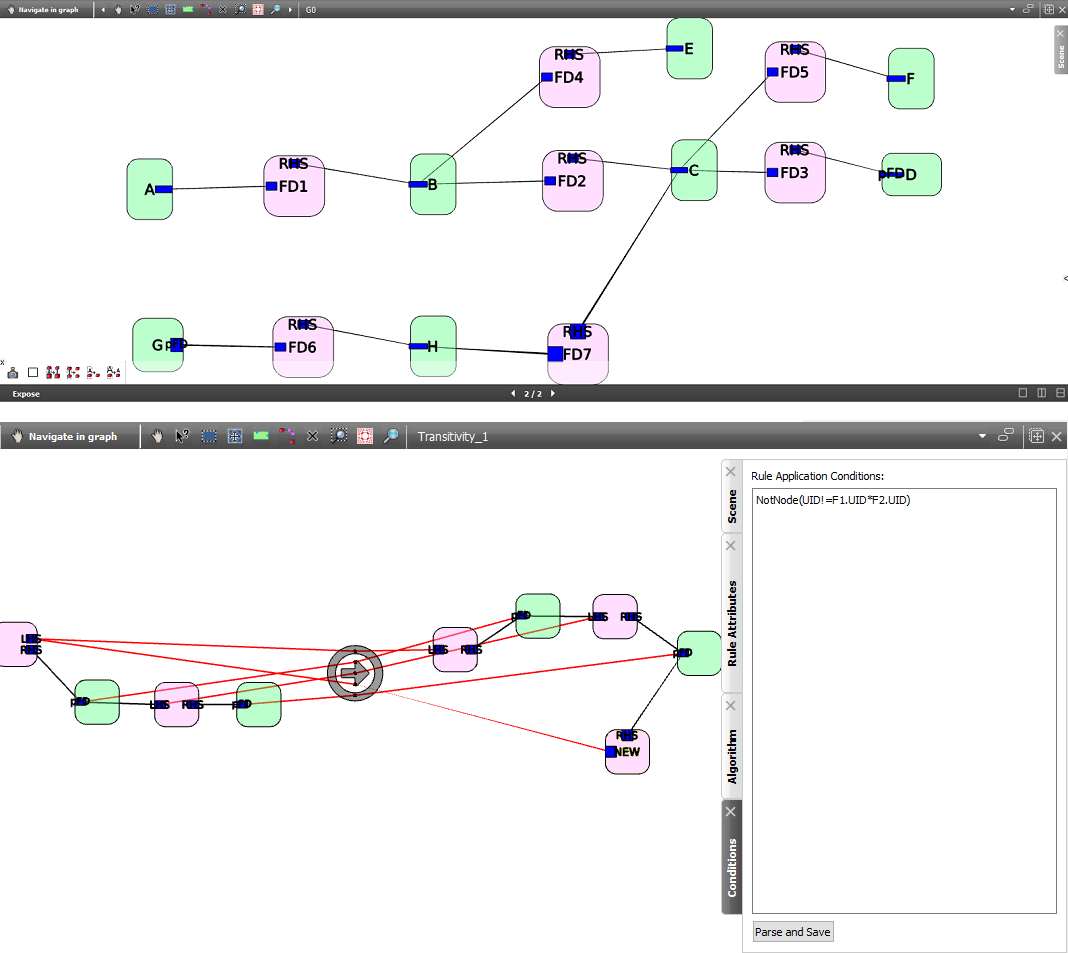}
    \caption{PORGY, an initial graph ($G0$), Transitivity\textsubscript{1} rule and the conditions editor.}
    \label{fig:PorgyScreenshot}
\end{figure}

\section{Transitive Closure Strategy}
\label{sec:tcstrat}
In this section we present the strategy to find the transitive closure of a set of FDs using the previously introduced visual language and rule conditions.

In PORGY, the starting point of port graph rewriting, the original model, is referred to as G0. In our case, G0 is a Functional Dependency Port Graph, as defined in Section~\ref{sec:portgraphdb}. The task is to find a relational transitive closure of G0, i.e. to generate all the transitive dependencies.

Inspired by the Chase Algorithm~\cite{AhoBU79,MaierMS79} we \emph{iterate} every FD in G0 (and also those added by the rules), apply the rules that detect a transitive dependency pattern and create the new FD. Once we have every possible new transitive dependency that goes through the iterated FD, we mark it \emph{visited}. We define two new Boolean node attributes: \iter to flag the node currently being iterated and \visit to permanently flag the node as visited. The two rules controlling the iteration are \emph{IterOn} and \emph{IterOff} (omitted). \emph{IterOn} randomly selects a node with attribute values {\tt RelDbType=FD}, {\tt iter=false}, {\tt visit=false} and sets the two flags: {\tt iter=true}, {\tt visit=true}. \emph{IterOff} rule selects the currently iterated FD node {\tt RelDbType=FD}, {\tt iter=true}, {\tt visit=true} and turns the iteration flag off: {\tt iter=false}, {\tt visit=true}.


\begin{figure}[!htb]
    \centering
    \begin{overpic}[scale=0.35]{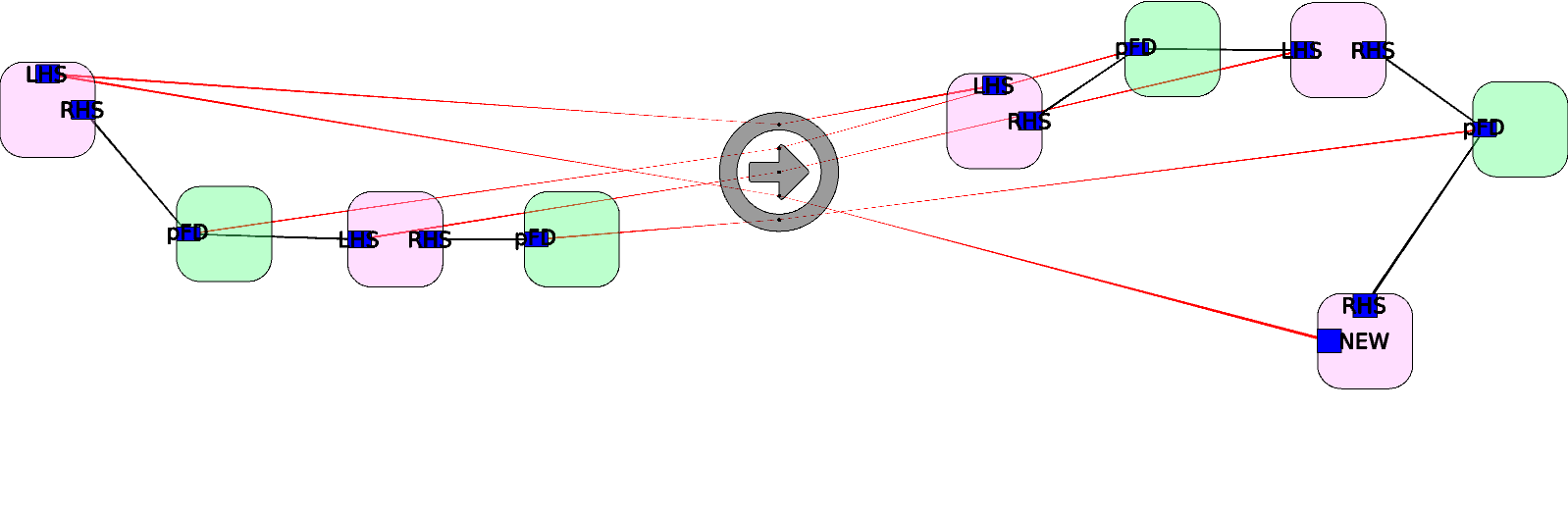}
        \texttt{
            \put(2,103){F2}
            \put(100,63){F1}
            \put(152,71){A}
            \put(60,71){B}
            \put(254,100){F2}
            \put(353,128){F1}
            \put(408,100){A}
            \put(315,120){B}            
            \put(350,20){{\parbox{2in} {\tiny{iter=F\\visit=F\\UID=F1.UID*F2.UID}}}}
            \put(0,90){{\parbox{2in} {\tiny{FDRHS.FunctionalArity=1}}}}
            \put(95,40){{\parbox{2in} {\tiny{FDLHS.FunctionalArity=1\\FDRHS.FunctionalArity=1\\iter=T\\visit=T}}}}
            \put(190,50){{\parbox{2in} {\tiny{\underline{Rule Condition:}\\NotNode(UID==F1.UID*F2.UID)}}}}
        }
    \end{overpic}
    \caption{Transitivity\textsubscript{1} rule.}
    \label{fig:RuleTransitivity1}
\end{figure}

\begin{figure}[!htb]
    \centering
    \begin{overpic}[scale=0.35]{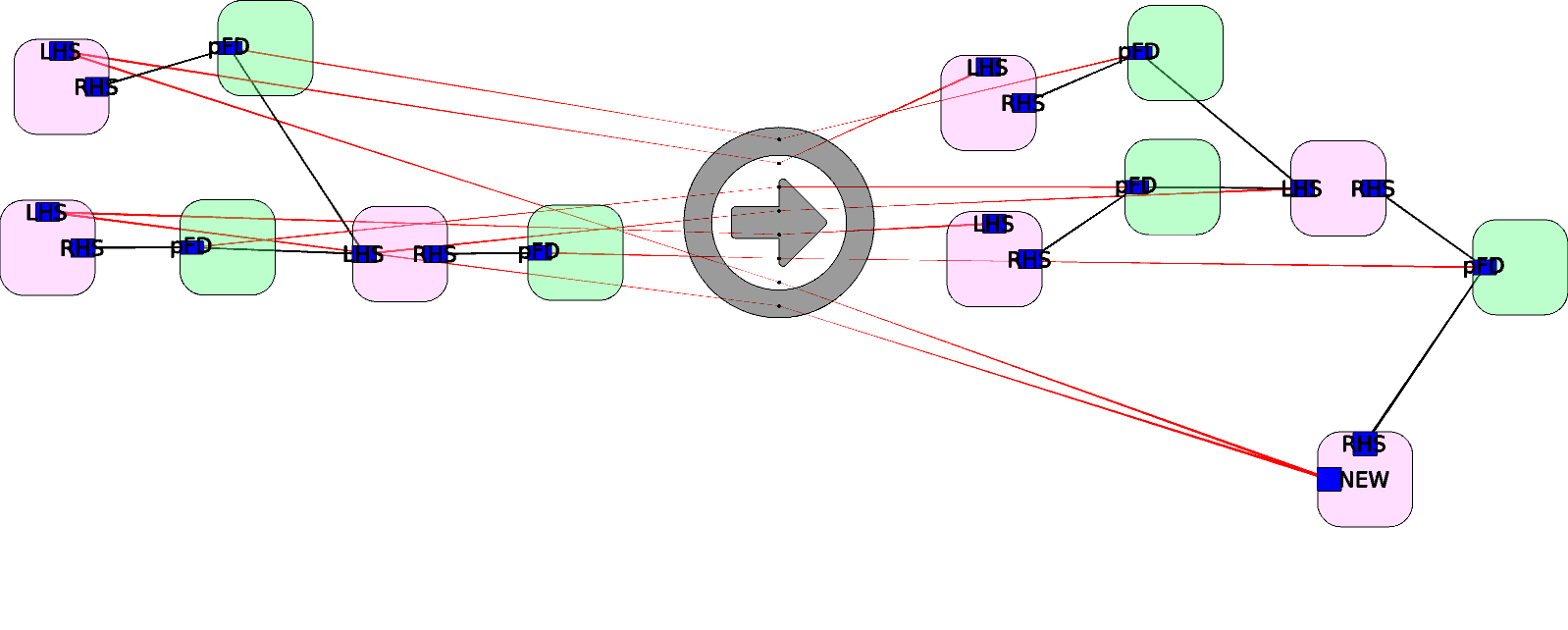}
        \texttt{
            \put(6,140){F2}
            \put(2,95){F3}            
            \put(100,90){F1}
            \put(70,155){C}
            \put(60,100){B}
            \put(155,100){A}
            \put(255,92){F3}
            \put(253,135){F2}
            \put(353,122){F1}
            \put(408,92){A}
            \put(315,115){B}
            \put(315,150){C}
            \put(345,15){{\parbox{2in} {\tiny{iter=F\\visit=F\\UID=F1.UID*F2.UID*F3.UID}}}}
            \put(0,80){{\parbox{2in} {\tiny{FDRHS.FunctionalArity=1}}}}
            \put(0,124){{\parbox{2in} {\tiny{FDRHS.FunctionalArity=1}}}}
            \put(95,65){{\parbox{2in} {\tiny{FDLHS.FunctionalArity=2\\FDRHS.FunctionalArity=1\\iter=T\\visit=T}}}}
            \put(180,50){{\parbox{2in} {\tiny{\underline{Rule Condition:}\\NotNode(UID==F1.UID*F2.UID*F3.UID)}}}}
        }
    \end{overpic}
    \caption{Transitivity\textsubscript{2} rule.}
    \label{fig:RuleTransitivity2}
\end{figure}

The first two rules we created can be seen on Figures~\ref{fig:RuleTransitivity1}~and~\ref{fig:RuleTransitivity2}. 
These rules detect FDPG-Paths representing transitive functional dependency chains e.g. $X \to Y \to A$ which, per Armstrong's Transitivity Axiom, means that $X \to A$ holds.
Red edges in the rules go through the \emph{bridge} arrow node ports. These edges mean that all edges from outside the matching subgraph into the port on the LHS end of the red edge are to be copied to connect to the port on the RHS end of the red edge. For example, the red edges running between the pFDLHS ports of nodes F2 in both rules 
and of node F3 in Transitivity\textsubscript{2} rule ensure that 
\begin{itemize}
	\item the pFDLHS port of the newly created dependency will be connected to all attributes that were on the LHS of F2 (and F3);
	\item the pFDLHS port of the new instance of F2 (and F3) will be connected to all attributes that were on the LHS of F2 (and F3) (the original dependency is preserved).
\end{itemize}

To offer extra, context-specific backtracking functionality, we assign prime number values to the UID attribute of every FD. When a new FD is created by any transitivity rule, the UID of the new FD is the product of the UIDs of the FDs that lead to the new FD. This is calculated and assigned by the Rule Algorithm feature of \Porgy. We use the Rule Condition functionality specified in Section~\ref{sec:rulecond} to control the applicability of the Transitivity rules. For example, the condition {\tt NotNode(UID==F1.UID*F2.UID)} means that if there is a node in the entire graph being rewritten (not only the matching subgraph) with a value in UID equal to the product of the UIDs of F1 and F2 (meaning the transitive dependency has already been found) then the rule shouldn't apply.

Finally, the strategy that uses the above defined rules is Strategy~\ref{alg:trancl}. The \emph{ResetVisitedFlags} rule sets the flags {\tt visit=F}, to allow for a whole new loop to run on all FD nodes.

\begin{algorithm}
    1. \hspace*{8pt} $\while{\tmatch{IterOn}} \doo{$\\
    a) \hspace*{16pt} $\one{IterOn}$;\\
    b) \hspace*{16pt} $\repeatt{\one{Transitivity_1}}$;\\
    c) \hspace*{16pt} $\one{IterOff}$\\
    \hspace*{12pt} };\\
    2. \hspace*{8pt} $\repeatt{\one{ResetVisitedFlags}}$;\\
    3. \hspace*{8pt} $\while{\tmatch{IterOn}} \doo{$\\
    a) \hspace*{16pt} $\one{IterOn}$;\\
    b) \hspace*{16pt} $\repeatt{\one{Transitivity_2}}$;\\
    c) \hspace*{16pt} $\one{IterOff}$\\
    \hspace*{12pt} }
    \caption{Transitive Closure Strategy}\label{alg:trancl}
\end{algorithm}

Note that the node being iterated (node name variable: F1) plays a pivotal role in the matching and its LHS subgraph increases in size as the \emph{FunctionalArity} of its FDLHS port increases.

The Transitive Closure Strategy is explained as follows:
\begin{enumerate}
	\item For as long as there is at least one FD node the strategy hasn't \emph{visited} and \emph{iterated}, do
	\begin{enumerate}
		\item Pick one such FD at random (with equal probabilites) and mark it \emph{visited} and \emph{iterated},
		\item Find and apply all possible applications of the rule \emph{Transitivity\textsubscript{1}} with the FD node picked in the previous step as the pivotal node F1,
		\item Mark F1 \emph{visited} but \emph{not iterated};
	\end{enumerate}
	\item Mark all FD nodes \emph{not visited};
	\item For as long as there is at least one FD node the strategy hasn't \emph{visited} and \emph{iterated}, do
	\begin{enumerate}
		\item Pick one such FD at random (with equal probabilites) and mark it \emph{visited} and \emph{iterated},
		\item Find and apply all possible applications of the rule \emph{Transitivity\textsubscript{2}} with the FD node picked in the previous step as the pivotal node F1,
		\item Mark F1 \emph{visited} but \emph{not iterated};
	\end{enumerate}
\end{enumerate}

\begin{property}[Derivation Tree of Strategy~\ref{alg:trancl}]
  \label{prop:derivtreebranch}
  Every non-leaf node in the Derivation Tree of Strategy~\ref{alg:trancl} has only one child node. Conversely, there is only one leaf node.
  \begin{proof}
    From the semantics of strategic port graph programs~\cite{FernandezKP18} we know that strategy constructs we use (\tmatch{}, \while{}, \repeatt{\one{}} and \one{})
	will not branch the Derivation Tree at all. The loop constructs we use execute their arguments sequentially as many times as they apply.
  \end{proof}	
\end{property}

Rules \emph{IterOn}, \emph{IterOff} and \emph{ResetVisitedFlags} are omitted due to space constraints. We discuss the semantic and rewriting characteristics of Strategy~\ref{alg:trancl} in Theorems~\ref{thm:tcsound}, \ref{thm:tccomplete} and \ref{thm:terminating}.

\begin{thm}
    \label{thm:tcsound}
	The Transitive Closure Strategy is sound. That is, it \textbf{only} finds functional dependencies that can be inferred from the original set of FDs 
	using Armstrong's Axioms but ignoring the meaningless dependencies that would be generated by the Reflexivity axiom.
\end{thm}
\begin{proof}
	Let $\Sigma$ be a set of functional dependencies and $G_\Sigma=(V,P,E,D)_{\F}$ the FDPG representing $\Sigma$. Let $K,L \subset V_A$ be attribute node sets representing 
	attribute sets $X, Y$ (resp.) and $n \in V_A$ an attribute node representing attribute $A$.
	We want to show that, after executing the Transitive Closure Strategy on $G_\Sigma$ obtaining $G_\Sigma'$, an FDPG-Path $\langle K,n \rangle$ exists in $G_\Sigma'$
	\textbf{only if} the functional dependency $X \to A$ can be inferred from $\Sigma$ using Armstrong's axioms.\\	
	It follows from Definition~\ref{def:fdpgpath} and Properties~\ref{prop:trans}~and~\ref{prop:union} that the existence of FDPG-Path $\langle X,Y,A \rangle$ 
	in a leaf node of the Derivation Tree means that either $\Sigma$ contained the explicit dependency $X \to A$ or that the Transitive Closure Strategy added it.
	Our rules and strategies (in fact, \Porgy itself) follow the rewriting principle that only rules can change the structure of the graph, strategy expressions can not. 
	We ignore rules \emph{IterOn} and \emph{IterOff} because they only affect values of node attributes \emph{iter} and \emph{visit} and not alter FDPG-Paths in the graph.
	This means we have to show that \emph{Transitivity\textsubscript{1..k}} rules create only valid FDPG-Paths.
	These rules detect FDPG-Paths $\langle K,L,n \rangle$, add a new Functional Dependency node to represent the transitive dependency $X \to A$, 
	and (using plain and red edges) ensure that	all LHS and RHS attributes are connected to the ports of the new dependency node. 
	This creates the direct path $\langle K,n \rangle$.
	According to Definitions~\ref{def:fdpg} and \ref{def:fdpgpath} this new path is valid.
\end{proof}

\begin{thm}
    \label{thm:tccomplete}
	The Transitive Closure Strategy is complete. That is, it finds \textbf{all} functional dependencies that can be inferred from the original set of FDs using Armstrong's Axioms but ignoring the meaningless dependencies that would be generated by the Reflexivity axiom.
\end{thm}
\begin{proof}
	We know that Armstrong's Reflexivity, Transitivity and Union rules, together, form a complete system of inference rules. 
	We have to show that these three rules can be deduced from the Transitive Closure Strategy, or more precisely, from rules \emph{Transitivity\textsubscript{1..k}}. 
	As noted earlier, a repeating pattern forms on the left-hand side of F1 in our Transitivity rules as $k$ increases.
	A $k$-ary Armstrong's Union operation is performed by \emph{Transitivity\textsubscript{k}}. For example, \emph{Transitivity\textsubscript{2}}: 
	if the left-hand sides of F2 and F3 are the same attribute set, and the rule simply copies those edges (which it does), 
	then the system behaves as if it performed the Union operation.	
	It is obvious from the structure of the FDPG-Rules that they perform what Armstrong's Transitivity rule does.
\end{proof}

\begin{thm}
    \label{thm:terminating}
    Let $\Sigma$ be a set of functional dependencies and $G_\Sigma=(V,P,E,D)_{\F}$ the FDPG representing $\Sigma$.
	The Transitive Closure Strategy (TCS) is a terminating program that never fails.
	That is, assuming that there are no cyclical dependencies in $\Sigma$, there is no infinite descending chain in the Derivation Tree.
\end{thm}
\begin{proof}
    \emph{Never fails.} We know from the semantics set out in~\cite{FernandezKP18} that the expressions $\whiledoo{C}{S}$ and $\repeatt{S}$ never fail. 
    $\one{IterOn}$ will also never fail because $\while{\tmatch{IterOn}}$ will check if it is possible to execute it at all.
    $\one{IterOff}$ will always apply because every execution of it is preceded by an execution of \emph{IterOn} and \emph{Transitivity\textsubscript{k}} rules do not change
    the \emph{iter} and \emph{visited} values of any pre-existing nodes.    
	
	\emph{Finitely branching.} It follows from Property~\ref{prop:derivtreebranch} that the Derivation Tree of the strategic graph program $(G_\Sigma$, TCS$)$ is finitely branching.
	
	\emph{Embedding.} We want to find an embedding of $(G_\Sigma$, TCS$)$ into $(\mathbb{N}, >)$. 
	We do this by defining a measure for every loop in the strategy.
	For $\repeatt{\one{Transitivity_k}}$, the measure is the number of possible matches of the LHS subgraph of F1. This is a good measure because 
	a) the Rule Condition on UID prevents re-application of the rule and b) even though the NEW node is added, it is not part of the LHS subgraph of F1.
	For $\whiledoo{\tmatch{IterOn}}{...}$ the measure is $|V_{FD}|_{G0} + |\Sigma^+| - |V_{FD}|_{Gi}$. 
	That is, the initial number of FD nodes in $G_\Sigma$ plus the size of the transitive closure of $\Sigma$ less the number of FD nodes after the $i$th application of the loop.
	With one successful application of \emph{Transitivity\textsubscript{k}} the number of FD nodes increases therefore the measure decreases.	
	Note that neither of these measures can be 0 or less. From these it follows that the above detailed measure provides a good monotone mapping into $(\mathbb{N}, >)$. 
	
	Then from Lemma~\ref{lemma:terminates} it follows that the strategy terminates.	
\end{proof}

\section{Conclusion}
\label{sec:conclusion}
Our results can be used to build a strategic graph program that takes the transitive closure as input and finds a minimal cover~\cite{AusielloDS83,Maier80}. From a minimal cover, all candidate keys of a relation schema can be found~\cite{SaiedianS96}. The minimal cover and the set of candidate keys can then be used as inputs to Bernstein's Synthesis Algorithm to synthesize Third Normal Form~\cite{Bernstein76} schemata.

\emph{Future Work}
Definition \ref{def:fdpg} may be generalized such that the Attribute nodes have a list of dependency ports, each element of the list corresponding to a type of data dependency that can exist in a database (e.g. multi-valued dependency). Converesely, each dependency type requires the introduction of its own Dependency node type. Each new Dependency node type may represent bilateral dependencies (LHS and RHS) or multilateral ones. Each new dependency type will require its own Path definition.

\bibliographystyle{eptcs}
\bibliography{Bibliography}

\begin{thebibliography}{10}
\providecommand{\bibitemdeclare}[2]{}
\providecommand{\surnamestart}{}
\providecommand{\surnameend}{}
\providecommand{\urlprefix}{Available at }
\providecommand{\url}[1]{\texttt{#1}}
\providecommand{\href}[2]{\texttt{#2}}
\providecommand{\urlalt}[2]{\href{#1}{#2}}
\providecommand{\doi}[1]{doi:\urlalt{http://dx.doi.org/#1}{#1}}
\providecommand{\bibinfo}[2]{#2}

\bibitemdeclare{article}{AhoBU79}
\bibitem{AhoBU79}
\bibinfo{author}{Alfred~V. \surnamestart Aho\surnameend},
  \bibinfo{author}{Catriel \surnamestart Beeri\surnameend} \&
  \bibinfo{author}{Jeffrey~D. \surnamestart Ullman\surnameend}
  (\bibinfo{year}{1979}): \emph{\bibinfo{title}{The Theory of Joins in
  Relational Databases}}.
\newblock {\sl \bibinfo{journal}{ACM Trans. Database Syst.}}
  \bibinfo{volume}{4}(\bibinfo{number}{3}), pp. \bibinfo{pages}{297--314}.
\newblock \urlprefix\url{http://doi.acm.org/10.1145/320083.320091}.

\bibitemdeclare{article}{AusielloDS83}
\bibitem{AusielloDS83}
\bibinfo{author}{Giorgio \surnamestart Ausiello\surnameend},
  \bibinfo{author}{Alessandro \surnamestart D'Atri\surnameend} \&
  \bibinfo{author}{Domenico \surnamestart Sacc{\`{a}}\surnameend}
  (\bibinfo{year}{1983}): \emph{\bibinfo{title}{Graph Algorithms for Functional
  Dependency Manipulation}}.
\newblock {\sl \bibinfo{journal}{J. {ACM}}}
  \bibinfo{volume}{30}(\bibinfo{number}{4}), pp. \bibinfo{pages}{752--766},
  \doi{10.1145/2157.322404}.

\bibitemdeclare{book}{termrewrite}
\bibitem{termrewrite}
\bibinfo{author}{Franz \surnamestart Baader\surnameend} \&
  \bibinfo{author}{Tobias \surnamestart Nipkow\surnameend}
  (\bibinfo{year}{1998}): \emph{\bibinfo{title}{Term rewriting and all that}}.
\newblock \bibinfo{publisher}{Cambridge University Press},
  \doi{10.1017/CBO9781139172752}.

\bibitemdeclare{article}{CallToArms}
\bibitem{CallToArms}
\bibinfo{author}{Antonio \surnamestart Badia\surnameend} \&
  \bibinfo{author}{Daniel \surnamestart Lemire\surnameend}
  (\bibinfo{year}{2011}): \emph{\bibinfo{title}{{A Call to Arms: Revisiting
  Database Design}}}.
\newblock {\sl \bibinfo{journal}{SIGMOD Rec.}}
  \bibinfo{volume}{40}(\bibinfo{number}{3}), pp. \bibinfo{pages}{61--69},
  \doi{10.1145/2070736.2070750}.

\bibitemdeclare{inproceedings}{BatiniD78}
\bibitem{BatiniD78}
\bibinfo{author}{Carlo \surnamestart Batini\surnameend} \&
  \bibinfo{author}{Alessandro \surnamestart D'Atri\surnameend}
  (\bibinfo{year}{1978}): \emph{\bibinfo{title}{Rewriting Systems as a Tool for
  Relational Data Base Design}}.
\newblock In \bibinfo{editor}{Volker \surnamestart Claus\surnameend},
  \bibinfo{editor}{Hartmut \surnamestart Ehrig\surnameend} \&
  \bibinfo{editor}{Grzegorz \surnamestart Rozenberg\surnameend}, editors: {\sl
  \bibinfo{booktitle}{Graph-Grammars and Their Application to Computer Science
  and Biology, International Workshop, Bad Honnef, October 30 - November 3,
  1978}}, {\sl \bibinfo{series}{Lecture Notes in Computer
  Science}}~\bibinfo{volume}{73}, \bibinfo{publisher}{Springer}, pp.
  \bibinfo{pages}{139--154}, \doi{10.1007/BFb0025717}.

\bibitemdeclare{inproceedings}{BeeriFH77}
\bibitem{BeeriFH77}
\bibinfo{author}{Catriel \surnamestart Beeri\surnameend},
  \bibinfo{author}{Ronald \surnamestart Fagin\surnameend} \&
  \bibinfo{author}{John~H. \surnamestart Howard\surnameend}
  (\bibinfo{year}{1977}): \emph{\bibinfo{title}{A Complete Axiomatization for
  Functional and Multivalued Dependencies in Database Relations}}.
\newblock In \bibinfo{editor}{Diane C.~P. \surnamestart Smith\surnameend},
  editor: {\sl \bibinfo{booktitle}{SIGMOD Conference}},
  \bibinfo{publisher}{ACM}, pp. \bibinfo{pages}{47--61}.
\newblock \urlprefix\url{http://doi.acm.org/10.1145/509404.509414}.

\bibitemdeclare{article}{Bernstein76}
\bibitem{Bernstein76}
\bibinfo{author}{Philip~A. \surnamestart Bernstein\surnameend}
  (\bibinfo{year}{1976}): \emph{\bibinfo{title}{Synthesizing Third Normal Form
  Relations from Functional Dependencies}}.
\newblock {\sl \bibinfo{journal}{ACM Trans. Database Syst.}}
  \bibinfo{volume}{1}(\bibinfo{number}{4}), pp. \bibinfo{pages}{277--298}.
\newblock \urlprefix\url{http://doi.acm.org/10.1145/320493.320489}.

\bibitemdeclare{article}{Codd70}
\bibitem{Codd70}
\bibinfo{author}{E.~F. \surnamestart Codd\surnameend} (\bibinfo{year}{1970}):
  \emph{\bibinfo{title}{{A Relational Model of Data for Large Shared Data
  Banks}}}.
\newblock {\sl \bibinfo{journal}{Communications of the ACM}}
  \bibinfo{volume}{13}(\bibinfo{number}{6}), pp. \bibinfo{pages}{377--387}.
\newblock \urlprefix\url{http://doi.acm.org/10.1145/362384.362685}.

\bibitemdeclare{inproceedings}{Codd71a}
\bibitem{Codd71a}
\bibinfo{author}{E.~F. \surnamestart Codd\surnameend} (\bibinfo{year}{1971}):
  \emph{\bibinfo{title}{Normalized Data Structure: A Brief Tutorial}}.
\newblock In \bibinfo{editor}{E.~F. \surnamestart Codd\surnameend} \&
  \bibinfo{editor}{A.~L. \surnamestart Dean\surnameend}, editors: {\sl
  \bibinfo{booktitle}{SIGFIDET Workshop}}, \bibinfo{publisher}{ACM}, pp.
  \bibinfo{pages}{1--17}, \doi{10.1145/1734714.1734716}.

\bibitemdeclare{incollection}{Courcelle90}
\bibitem{Courcelle90}
\bibinfo{author}{Bruno \surnamestart Courcelle\surnameend}
  (\bibinfo{year}{1990}): \emph{\bibinfo{title}{Graph Rewriting: An Algebraic
  and Logic Approach}}.
\newblock In: {\sl \bibinfo{booktitle}{Handbook of Theoretical Computer
  Science, Volume {B:} Formal Models and Sematics {(B)}}},
  \bibinfo{publisher}{MIT Press}, pp. \bibinfo{pages}{193--242}.
\newblock
  \urlprefix\url{https://www.elsevier.com/books/formal-models-and-semantics/unknown/978-0-444-88074-1}.

\bibitemdeclare{incollection}{EmbleyM11}
\bibitem{EmbleyM11}
\bibinfo{author}{David~W. \surnamestart Embley\surnameend} \&
  \bibinfo{author}{W.~Y. \surnamestart Mok\surnameend} (\bibinfo{year}{2011}):
  \emph{\bibinfo{title}{{Mapping Conceptual Models to Database Schemas}}}.
\newblock \bibinfo{volume}{XIX}, \bibinfo{publisher}{Springer}, pp.
  \bibinfo{pages}{123--164}.
\newblock
  \urlprefix\url{http://www.springer.com/computer/swe/book/978-3-642-15864-3}.

\bibitemdeclare{article}{FernandezKP18}
\bibitem{FernandezKP18}
\bibinfo{author}{Maribel \surnamestart Fern{\'a}ndez\surnameend},
  \bibinfo{author}{H{\'e}l{\`e}ne \surnamestart Kirchner\surnameend} \&
  \bibinfo{author}{Bruno \surnamestart Pinaud\surnameend}
  (\bibinfo{year}{2018}): \emph{\bibinfo{title}{{Strategic Port Graph
  Rewriting: an Interactive Modelling Framework}}}.
\newblock {\sl \bibinfo{journal}{{Mathematical Structures in Computer
  Science}}}, pp. \bibinfo{pages}{1--48}, \doi{10.1017/S0960129518000270}.
\newblock \urlprefix\url{https://hal.inria.fr/hal-01251871}.

\bibitemdeclare{incollection}{JahnkeZundorf99}
\bibitem{JahnkeZundorf99}
\bibinfo{author}{J.~H. \surnamestart Jahnke\surnameend} \&
  \bibinfo{author}{A.~\surnamestart Z{\"u}ndorf\surnameend}
  (\bibinfo{year}{1999}): \emph{\bibinfo{title}{{Applying Graph Transformations
  to Database re-engineering}}}.
\newblock \bibinfo{volume}{2}, \bibinfo{publisher}{World Scientific}, pp.
  \bibinfo{pages}{267--286}.
\newblock
  \urlprefix\url{http://www.worldscientific.com/worldscibooks/10.1142/4180}.

\bibitemdeclare{article}{Maier80}
\bibitem{Maier80}
\bibinfo{author}{David \surnamestart Maier\surnameend} (\bibinfo{year}{1980}):
  \emph{\bibinfo{title}{{Minimum Covers in Relational Database Model}}}.
\newblock {\sl \bibinfo{journal}{J. {ACM}}}
  \bibinfo{volume}{27}(\bibinfo{number}{4}), pp. \bibinfo{pages}{664--674},
  \doi{10.1145/322217.322223}.

\bibitemdeclare{article}{MaierMS79}
\bibitem{MaierMS79}
\bibinfo{author}{David \surnamestart Maier\surnameend},
  \bibinfo{author}{Alberto~O. \surnamestart Mendelzon\surnameend} \&
  \bibinfo{author}{Yehoshua \surnamestart Sagiv\surnameend}
  (\bibinfo{year}{1979}): \emph{\bibinfo{title}{Testing Implications of Data
  Dependencies}}.
\newblock {\sl \bibinfo{journal}{{ACM} Trans. Database Syst.}}
  \bibinfo{volume}{4}(\bibinfo{number}{4}), pp. \bibinfo{pages}{455--469},
  \doi{10.1145/320107.320115}.

\bibitemdeclare{article}{SaiedianS96}
\bibitem{SaiedianS96}
\bibinfo{author}{Hossein \surnamestart Saiedian\surnameend} \&
  \bibinfo{author}{Thomas \surnamestart Spencer\surnameend}
  (\bibinfo{year}{1996}): \emph{\bibinfo{title}{An Efficient Algorithm to
  Compute the Candidate Keys of a Relational Database Schema}}.
\newblock {\sl \bibinfo{journal}{Comput. J.}}
  \bibinfo{volume}{39}(\bibinfo{number}{2}), pp. \bibinfo{pages}{124--132},
  \doi{10.1093/comjnl/39.2.124}.

\end{thebibliography}

\end{document}